\documentclass[11pt]{article}

\usepackage[english]{babel}
\usepackage{amsmath}
\usepackage{amssymb}
\usepackage{amstext,amsthm,bbm}
\usepackage{graphicx}
\usepackage{color}

\usepackage{enumerate}
\usepackage{fullpage}

\usepackage{url}
\usepackage{amssymb}
\usepackage{bm}
\usepackage{amsfonts}
\usepackage{indentfirst}
\usepackage[toc,title]{appendix}
\usepackage[font=small]{caption}
\usepackage{cite}
\usepackage{xcolor}
\usepackage{amsthm}

\newtheorem{theorem}{Theorem}[section]
\newtheorem{corollary}[theorem]{Corollary}

\newtheorem{lemma}[theorem]{Lemma}

\theoremstyle{definition}

\newtheorem{remark}[theorem]{Remark}

\newcommand{\tr}{\textup{Tr}}

\newcommand{\<}{\left<}
\renewcommand{\>}{\right>}
\newcommand{\idm}{\mathbf{1}}

\newcommand{\cP}{{\cal P}}

\newcommand{\cL}{{\cal L}}

\newcommand{\bb}[1]{{\mathbf #1}}

\newcommand{\He}{\mathrm{He}}
\newcommand{\Gmm}[1]{\Gamma\left( #1 \right)}

\begin{document}

\title{Condition numbers for real eigenvalues in the real Elliptic Gaussian ensemble}
\author{Yan V. Fyodorov$^1$\footnote{yan.fyodorov@kcl.ac.uk} \  and Wojciech Tarnowski$^2$\footnote{wojciech.tarnowski@doctoral.uj.edu.pl}
\\[0.5ex] {\small $^1$Department of Mathematics, King's College London}
{\small Strand, London, WC2R 2LS, UK}
\\
 {\small $^2$Institute of Theoretical Physics,} \\
\small{ Jagiellonian University, S. \L ojasiewicza 11, PL 30-348 Krak\'ow, Poland}
}

\date{\today}
\maketitle
\begin{abstract}
We study the distribution of the eigenvalue condition numbers $\kappa_i=\sqrt{ (\bb{l}_i^* \bb{l}_i)(\bb{r}_i^* \bb{r}_i)}$
associated with real eigenvalues $\lambda_i$ of partially asymmetric $N\times N$ random matrices from the real Elliptic Gaussian  ensemble.
The large values of $\kappa_i$ signal  the non-orthogonality of the (bi-orthogonal) set of left $\bb{l}_i$ and right $\bb{r}_i$ eigenvectors and enhanced sensitivity of the associated eigenvalues against perturbations of the matrix entries. We derive the general finite $N$ expression for the joint density function(JDF) ${\cal P}_N(z,t)$ of $t=\kappa_i^2-1$ and $\lambda_i$ taking value $z$, and investigate its several scaling regimes in the limit $N\to \infty$. When the degree of asymmetry is fixed as $N\to \infty$, the number of real eigenvalues is $O(\sqrt{N})$, and in the bulk of the real spectrum $t_i=O(N)$, while on approaching the spectral edges the non-orthogonality is weaker: $t_i=O(\sqrt{N})$. In both cases the corresponding JDFs, after appropriate rescaling, coincide with those found in the earlier studied case of fully asymmetric (Ginibre) matrices. A different regime of weak asymmetry arises when a finite fraction of $N$ eigenvalues remain real as $N\to \infty$. In such a regime eigenvectors are weakly non-orthogonal, $t=O(1)$, and we derive the associated JDF, finding  that the characteristic tail ${\cal P}(z,t)\sim t^{-2}$  survives for arbitrary weak asymmetry. As such, it is the most robust feature of the condition number density for real eigenvalues of asymmetric matrices.
\end{abstract}

\begin{small}
\textbf{Keywords:} bi-orthogonal eigenvectors, eigenvalue condition numbers, weak non-Hermiticity

\textit{Mathematics Subject Classification:}  60B20  	Random matrices 
\end{small}

\section{Introduction}
A (real-valued) square matrix $X$ is asymmetric if it is different from its transpose $X^T$, and non-normal if $XX^T\neq X^TX$. Generically, asymmetric matrices are non-normal, and their eigenvalues are much more sensitive to the perturbations of the matrix entries than for their symmetric (hence selfadjoint and normal) counterparts. It is well-known that non-normality may raise serious issues when calculating the spectra of such matrices numerically: keeping a fixed precision of calculations might not be sufficient, as some eigenvalues can be `ill-conditioned'.

  To be more specific, we assume that $X$ can be diagonalized (which for random matrices happens with probability one). Then to each eigenvalue $\lambda_i$, real or complex (in the latter case being always accompanied by its complex conjugate partner $\overline{\lambda}_i$) correspond two sets of eigenvectors, left $\bb{l}_i$ and right $\bb{r}_i$ which can always  be chosen to be {\it bi-}orthogonal: $\bb{l}_i^* \bb{r}_{j}=\delta_{ij}$, where $\bb{l}_i^*:=\overline{\bb{l}_i^T}$ stands for Hermitian conjugation.
   The corresponding eigenproblems are  $X\bb{r}_i=\lambda_i\bb{r}_i$ and $X^T \bb{l}_i=\lambda_i \bb{l}_i$.  Consider now a matrix $X'=X+\epsilon P$, where the second term represents an error one makes by storing the matrix entries with a finite precision, with $\epsilon>0$ controlling the magnitude of the error and $P$ reflecting the matrix structure of the perturbation. In the first order perturbation theory in parameter $\epsilon$ eigenvalues are shifted by
\begin{equation}
|\delta \lambda_i|=\epsilon |\bb{l}_i^* P \bb{r}_i|\leq \epsilon ||P|| \sqrt{ (\bb{l}_i^* \bb{l}_i)(\bb{r}_i^* \bb{r}_i)}.
\end{equation}
The latter quantity, $\kappa_i=\sqrt{ (\bb{l}_i^* \bb{l}_i)(\bb{r}_i^* \bb{r}_i)}$, shows that the sensitivity of eigenvalues is essentially controlled by the non-orthogonality of the corresponding left and right eigenvectors. Correspondingly, in the numerical analysis context $\kappa_i$ is called the eigenvalue condition number of the eigenvalue $\lambda_i$ \cite{W1965, TE2005}. Note also that the Cauchy-Schwarz inequality implies $\kappa\geq 1$, with $\kappa=1$ only  when $X$ is normal.

It is natural to ask how well-conditioned  a `typical' asymmetric matrix is. This question can be most meaningfully answered in the context of Random Matrix Theory (RMT), by defining `typical' as randomly chosen according to a probability measure specified by a particular choice of the ensemble. The simplest yet nontrivial choice is to assume that all entries are mean-zero, independent, identically distributed (i.i.d.) Gaussian numbers. This defines the standard real Ginibre ensemble which we denote $Gin_1$. Note that the question is equally interesting for matrices with entries that are complex rather than real, defining the complex Ginibre ensemble, which we denote $Gin_2$.
Note that for such an ensemble eigenvalues $\lambda_i$ are purely complex with probability one.

  It is the latter ensemble for which the study of the eigenvalue condition numbers has been initiated two decades ago by Chalker and Mehlig~\cite{CM1998,MC2000}.  More precisely, Chalker and Mehlig introduced a matrix of inner products $O_{ij}=(\bb{l}_i^*\bb{l}_j)(\bb{r}_j^* \bb{r}_i)$, which they called  `eigenvector overlaps'. The diagonal elements of that matrix are simply the squared eigenvalue condition numbers. They further associated with the diagonal elements of the overlap matrix the following single-point correlation function:
\begin{equation}
O_1(z)=\left\langle\frac{1}{N}\sum_{k=1}^{N}O_{kk}\delta(z-\lambda_k) \right\rangle_{Gin_2}. \label{eq:O1}
\end{equation}
where the angular brackets stand for the expectation with respect to the probability measure associated with the complex Ginibre ensemble,
and $\delta(z-\lambda_k)$ is the Dirac delta mass at the eigenvalue $\lambda_k$, so that the empirical density of eigenvalues in the complex plane $z$ reads $\rho(z)=\frac{1}{N}\sum_{k=1}^{N}\delta(z-\lambda_k).$

 Such $O_1(z)$ gives the conditional expectations of (squared) $\kappa$ as $\mathbb{E}(\kappa_i^2|z=\lambda_i)=\frac{O_1(z)}{\langle \rho(z)\rangle }$, where $\langle \rho(z)\rangle$ is the mean spectral density around $z$~\cite{BNST2017}.
It turned out that in the bulk of the spectrum of the complex Ginibre ensemble the magnitude of a typical diagonal overlap $O_{ii}$ grows linearly with the size of the matrix $N$, so one needs to consider a rescaled object $\tilde{O_1}(z)=\frac{1}{N}O_1(z)$ to obtain a non-trivial limit. In their influential papers~\cite{CM1998,MC2000} Chalker and Mehlig used the `formal' perturbation theory expansion to evaluate asymptotically, for $N\gg 1$, both the diagonal overlap $O_1(z)$ and its more general off-diagonal counterpart
$$
O_2(z_1,z_2)=\left\langle  \frac{1}{N}\sum_{k\ne l}^N
\mathcal{O}_{kl}\delta(z_1-\lambda_k)\delta(z_2-\lambda_l)\right\rangle_{Gin_2}\,.
$$
 The first mathematically rigorous verification of the Chalker and Mehlig result for the diagonal overlap has been done in \cite{WaltersStarr2015}. Remarkably, the function $O_1(z)$  can be efficiently studied within the formalism of free probability~\cite{Jetal1999} which  recently allowed to extend the Chalker-Mehlig formulas to a broad class of invariant ensembles beyond the Gaussian case~\cite{BNST2017,NT2018}. $O_1(z)$ is also known for a finite size of the complex Ginibre matrix~\cite{CM1998,MC2000} and products of small Ginibre matrices~\cite{BSV2017}.   It has been recently shown that for complex Ginibre matrices the one- and two-point functions conditioned on an arbitrary number of eigenvalues are related to determinantal point processes~\cite{ATTZ2019}.
 Various features characterising the rich properties of  eigenvectors of nonnormal random matrices have been also 
 studied in \cite{BZ2018} and \cite{CR2018}. 

Here it is necessary to mention that the interest in statistical properties of the overlap matrix $O_{kl}$ and related objects extends much beyond the issues of eigenvalue stability under perturbation and is driven by numerous applications in theoretical and mathematical physics. In particular, non-orthogonality governs transient dynamics in complex systems~\cite{G2017,GNetal2018,TNV2019} (see also~\cite{MBO2018,EKR2019}),  analysis of spectral outliers in non-selfadjoint matrices~\cite{MetzNeri}, and,  last but not least, the description of the Dyson Brownian motion for non-normal matrices ~\cite{BGNTW2014,BD2018,GW2018}. Another steady source of interest in
the statistics of eigenvector overlaps is due to its role in chaotic wave scattering. In that context $O_1(z)$ and $O_2(z_1,z_2)$ have been studied for a few special models different from Ginibre (both theoretically \cite{FS2000,FM2002,FS2012} and very recently experimentally \cite{DavyGenack18,DavyGenack19}) and  in the associated models of random lasing~\cite{PSB2000,SFPB2000}.
In the scattering context all eigenvalues are necessarily complex, and the lasing threshold is associated with the eigenvalue with the smallest imaginary part. For that special eigenvalue even the distribution of the overlap $O_{ii}$ has been studied~\cite{SFPB2000}.

  In fact,  Chalker and Mehlig not only analysed $O_1(z)$, but also put forward a conjecture on the tail for the density  ${\cal P}\left(\mathcal{O}_{ii}\right)$  of the distribution
  of diagonal overlaps $\mathcal{O}_{ii}$. Namely, based on the exactly solvable case of $2\times 2$ matrices and numerical simulations
 for the complex Ginibre case they predicted that for large overlaps the density will exhibit a tail ${\cal P}\left(\mathcal{O}_{ii}\right)\sim \mathcal{O}^{-3}_{ii}$, making all the positive integer moments beyond  $O_1(z)$ divergent. This conjecture
    has been settled only recently  with two different methods, by Bourgade and Dubach in ~\cite{BD2018} (where some information about
$\mathcal{O}_{l\ne k}$ was also provided)  and by Fyodorov~\cite{F2018}.  The latter paper also revealed that for real eigenvalues
of real Ginibre matrices $Gin_1$ the diagonal overlaps $\mathcal{O}_{ii}$ are distributed with even heavier tail: ${\cal P}\left(\mathcal{O}_{ii}\right)\sim \mathcal{O}^{-2}_{ii}$, making even the mean of the overlap divergent.

To address the above distributions it is convenient to introduce the following natural generalization of the equation (\ref{eq:O1}):
\begin{equation}
{\cal P}_N(z,t)=\left\langle\sum_{i}^{}\delta(O_{ii}-1-t)\delta(z-\lambda_i)\right\rangle \label{eq:jointPdfDef}
\end{equation}
interpreted as the joint density function (JDF) of the `diagonal' (or `self-overlap') non-orthogonality  factor $t=O_{ii}-1$ for  the right and left eigenvectors corresponding to eigenvalues in the vicinity of a point $z=x+iy$ in the complex plane. The summation runs over all real eigenvalues for the real ensembles and over all eigenvalues for complex ensembles. As such, it is not a probability density, because it is normalized to the mean total number of (real) eigenvalues.
As was shown in \cite{BD2018,D2019a,F2018} the JDF ${\cal P}_N(z,t)$ tends (after an appropriate rescaling of the variables $z$ and $t$ with the size $N$) to the inverse gamma distribution as $N\gg 1$:
 \begin{equation}
 \lim_{N\to\infty} N\,{\cal P}_N(z\sqrt{N},Nt)=\frac{ \langle \rho(z)\rangle}{t} e^{-\frac{\tilde{O}_1(z)}{t\langle \rho(z)\rangle}}\left(\frac{\tilde{O}_1(z)}{t\langle \rho(z)\rangle}\right)^{\beta}, \quad |z|<1. \label{eq:BulkLimit}
 \end{equation}
Here parameter $\beta=1$ corresponds to the real eigenvalues of real Ginibre matrices (in which case the parameter $z$ should be chosen to be real), $\beta=2$ to the complex Ginibre case and $\beta=4$ to the quaternionic Ginibre case. Recall that in \eqref{eq:BulkLimit}
the limiting spectral density of real eigenvalues for $\beta=1$ is $\langle \rho(z)\rangle=\frac{1}{2\sqrt{2\pi}}$ for the interval
 $|z|<1$, whereas the limiting spectral density of complex eigenvalues for $\beta=2$ is $\langle \rho(z)\rangle=\pi^{-1}$ inside the unit
circle  $|z|<1$.

  For complex Ginibre the limiting expression (\ref{eq:BulkLimit}) naturally incorporates  the Chalker-Mehlig result. In the formula above $\tilde{O}(z)=\pi^{-1}(1-|z|^2)$, which is the large $N$ limit of the rescaled one-point correlation function. Interestingly, despite the fact that for $\beta=1$ the mean value $O_1(z)$ defined via \eqref{eq:O1} does not exist, the closely related rescaled combination, $\tilde{O}_1(z)=\frac{1}{2\sqrt{2\pi}}(1-z^2)$, appears as a parameter in the inverse $\gamma_1$ distribution and therefore defines 
  the {\it typical} value of the diagonal overlap. Further calculations in a few non-Gaussian rotationally-invariant matrix ensembles (in particular, associated with `truncations' of unitary matrices) done very recently in ~\cite{D2019} suggest that \eqref{eq:BulkLimit} might exhibit a certain degree of universality. Note that the statistics of $O_{ii}$ for complex eigenvalues of real Ginibre matrices remains an outstanding problem, though it would be natural to expect that also in that case, for a fixed $z$ with a non-vanishing imaginary part, the limit should be the same as for the complex Ginibre case.

Returning to the original question of eigenvalue condition numbers for real-valued matrices, the above results in particular imply that in contrast to well-conditioned eigenvalues of symmetric matrices with $\kappa=1$ the typical condition numbers in fully asymmetric random matrices grow with matrix size as $\sqrt{N}$~\cite{F2018} and show strong fluctuations. One of natural questions is then to ask how those properties evolve for matrices with a controlled degree of asymmetry in their entries. The aim of this work is to answer this question. To this end  we consider matrices with i.i.d. real Gaussian entries, such that the entries $X_{ij}$ and $X_{ji}$ are correlated. The joint pdf for the elements of this ensemble, known in the literature either as the real partly-symmetric Ginibre ensemble, or, alternatively, as the real Elliptic Gaussian ensemble, is given by
\begin{equation}
P(X)dX=C_{N}^{-1} \exp\left[-\frac{1}{2(1-\tau^2)}\tr (XX^T-\tau X^2)\right] dX. \label{eq:EllipticPDF}
\end{equation}
Here $dX=\prod_{i,j=1}^{N}dX_{ij}$ is the flat Lebesgue measure over all matrix elements and the normalization constant reads $C_N=(2\pi)^{N^2/2}(1+\tau)^{N/2}(1-\tau^2)^{\frac{N(N-1)}{4}}$. The parameter $\tau\in[0,1]$ tunes the degree of correlation, $\mathbb{E}(X_{ij}X_{ji})=\tau$ for $i\neq j$, and \eqref{eq:EllipticPDF} interpolates between the Real Ginibre Ensemble for $\tau=0$ and an ensemble of
real symmetric matrices (Gaussian Orthogonal Ensemble) for $\tau=1$.  In particularly, it is well-known that for large sizes $N\gg 1$ a nontrivial scaling regime of {\it weak non-Hermiticity} arises as long as the product $N(1-\tau)$ is kept of the order of unity \cite{FKS1997a,FKS1997b,FKS1998,Efetov1997,FS2003}. It is this regime when non-normality gradually develops, and the condition numbers $\kappa_i$ start growing away from the value $\kappa_i=1$. Our considerations allow us to address this regime in a quantitative way.

\section{Statement of the main results}
It turns out that the method of evaluating the JDF in (\ref{eq:jointPdfDef}) suggested for the Ginibre case in \cite{F2018} works for the Elliptic ensemble as well, though actual calculations turn out to be significantly more involved.  Relegating technical detail to the rest of the paper, in this section we present our main findings.

Our main theorem gives the joint density function of the eigenvalue $\lambda_i$ and the shifted overlap $t=O_{ii}-1$
for elliptic matrices of a given size $N$ distributed according to (\ref{eq:EllipticPDF}).
It takes a more compact form when the rescaled variable $q=(1-\tau)t$ is considered. Therefore, we define ${\cal P}_N^{\tau}(z,q) = (1-\tau)^{-1} {\cal P}_{N}(z,\frac{q}{1-\tau})$.

\begin{theorem} \label{th:Main}
Let $X_N$ be an $N\times N$ random matrix with the probability density function  given by \eqref{eq:EllipticPDF}. Let us define the rescaled and shifted eigenvalue condition number $q=(1-\tau)(\kappa^2-1)$. The joint density \eqref{eq:jointPdfDef} of a real eigenvalue $z$ and the associated squared condition number expressed via the variable $q$ is given by
{\small
\begin{multline} \label{eq:MainResult}
\cP_N^{\tau}(z,q)=\frac{1}{2(1+\tau)\sqrt{2\pi} \Gmm{N-1}} \frac{e^{-\frac{z^2}{2(1+\tau)}\left(1+\frac{q}{1+q}\right)}}{\sqrt{q(1+q)}}\left(\frac{q}{q+1+\tau}\right)^{\frac{N}{2}-1}\times
\\
\left[\frac{(1+\tau-2z^2)P_{N-2}+2z[R_{N-2}+\tau(N-2)R_{N-3}]}{1+q}+\frac{P_{N-2}z^2}{(1+q)^2}+\frac{\tau^2(1+\tau)^2N(N-2)P_{N-3}}{(1+\tau+q)^2}+ \right.
\\
\left.\frac{(1+\tau)(1-\tau^2)(N-2)((N-2)P_{N-3}-T_{N-3})}{1+\tau+q} -\frac{2\tau(1+\tau)(N-2)zR_{N-3}}{(1+q)(1+\tau+q)}\right],
\end{multline}
}
where the functions: $P_m:=P_m(z), \, R_m:=R_m(z), \, T_m:=T_m(z)$ are defined in terms of the Hermite polynomials
\begin{equation}
\He_m(z)=\frac{(\pm i)^m}{\sqrt{2\pi}} e^{\frac{z^2}{2}}\int_{\mathbb{R}}t^m e^{-\frac{t^2}{2}\mp i z t}dt,\label{eq:Hermite}
\end{equation}
as
{\small
\begin{eqnarray}
P_{N}(z)&=&N!\sum_{k=0}^{N}\frac{\tau^k}{k!}\left( (k+1)\He_k^2\left(\frac{z}{\sqrt{\tau}}\right)-k\He_{k-1}\left(\frac{z}{\sqrt{\tau}}\right)\He_{k+1}\left(\frac{z}{\sqrt{\tau}}\right)\right), \label{eq:PNdef}
\\
R_N(z) &=& \frac{N!}{2}\sum_{k=0}^{N}\frac{\tau^{k+\frac{1}{2}}}{k!}\left((k+2)\He_{k+1}\left(\frac{z}{\sqrt{\tau}}\right)\He_{k}
\left(\frac{z}{\sqrt{\tau}}\right)
\label{eq:RNdef}
 -k\He_{k+2}\left(\frac{z}{\sqrt{\tau}}\right)\He_{k-1}\left(\frac{z}{\sqrt{\tau}}\right)\right), 
\\
T_N(z) &=& N!\sum_{k=0}^{N}\frac{k \tau^k}{k!}\left( (k+1)\He_k^2\left(\frac{z}{\sqrt{\tau}}\right)-k\He_{k-1}\left(\frac{z}{\sqrt{\tau}}\right)\He_{k+1}\left(\frac{z}{\sqrt{\tau}}\right)\right). \label{eq:TNdef}
\end{eqnarray}
}
\end{theorem}
\remark Note that for $\tau=0$ these quantities simplify to $P_N=e^{z^2}\Gmm{N+1,z^2}$, $R_N=zP_N$ and $T_N=Nz^2P_{N-1}$, with $\Gmm{N+1,z}=\int_z^{\infty}u^Ne^{-x}du$, and the known result~\cite[eq. 2.5]{F2018} is recovered.
\begin{remark}
The exact mean density of purely real eigenvalues $ \rho^{(r)}_N(z)$ for the real Elliptic ensemble of even size $N$ is known thanks to Forrester and Nagao \cite{ForNag2008}. It is given by $\rho_N^{(r)}(z)=\rho_N^{(1)}(z)+\rho_N^{(2)}(z)$ with
\begin{align}
\rho_{N}^{(1)}(z)&=\frac{1}{\sqrt{2\pi}}e^{-\frac{z^2}{1+\tau}}\sum_{k=0}^{N-2}\frac{\tau^k}{k!}He_k^2\left(\frac{z}{\sqrt{\tau}}\right), \label{eq:RealDens1}
\\
\rho_{N}^{(2)}(z)&=\frac{1}{\sqrt{2\pi}(1+\tau)\Gmm{N-1}}e^{-\frac{z^2}{2(1+\tau)}}\tau^{N-3/2}He_{N-1}\left(\frac{z}{\sqrt{\tau}}\right)\int_0^z e^{-\frac{u^2}{2(1+\tau)}}He_{N-2}\left(\frac{u}{\sqrt{\tau}}\right)du. \label{eq:RealDens2}
\end{align}
For odd $N$  the density can be obtained using the method from~\cite{FM2009}.
 Our expression (\ref{eq:MainResult}) by its very definition must reproduce the Forrester-Nagao
result after integration over the variable $t$.  Performing such an integration analytically is, however, a challenging task which
we managed to complete for $N=2,3,4$. Nevertheless, we checked that performing such an integral numerically for moderate values of $N$ gives  results that are indistinguishable from the density of real eigenvalues, see Appendix~\ref{sec:AppDens}.
\end{remark}

As the expression (\ref{eq:MainResult}) is exact, the joint density in the original variable, $\cP_N(z,t)$, can be further analyzed in interesting scaling limits as $N\to \infty$.
The first of such limits is the so-called `bulk scaling' corresponding to the eigenvalues
inside the limiting support of the spectrum which (after appropriate rescaling $z\to \sqrt{N}z$) for a fixed $z$ and $0\le \tau<1$ represents an ellipse in the complex plane (hence the name for the ensemble), centered at the origin, with semi-axis $1-\tau$ along imaginary axis and $1+\tau$ along the real axis.
Since we are dealing only with real eigenvalues, we restrict ourselves to real $z$ such that $|z|<1+\tau$, where the following asymptotics holds:
\begin{corollary}\label{th:BulkScaling} (bulk scaling) Define for a fixed $0\le \tau<1$ and real $z$ satisfying  $|z|<1+\tau$ the limiting scaled JDF as
 ${\cal P}_{bulk}(z,t)=\lim_{N\to\infty}N{\cal P}_N (\sqrt{N}z,Nt)$. Then
\begin{equation}
{\cal P}_{bulk}(z,t)=\frac{\sqrt{1-\tau^2}}{2\sqrt{2\pi}}\frac{\left[1-\frac{z^2}{(1+\tau)^2}\right]}{t^2}e^{-\frac{1-\tau^2}{2t}\left[1-\frac{z^2}{(1+\tau)^2}\right]}.  \label{eq:BulkScaling}
\end{equation}
\end{corollary}
This asymptotics shows that the typical value of the diagonal overlap $t=O_{ii}-1$ in this regime is always of the order $N$ as $N\gg1$,
similarly to the behaviour for the Ginibre case $\tau=0$. Moreover, by recalling that the asymptotic density of real eigenvalues in  the elliptic case is $\left\langle \rho(z)\right\rangle=\frac{1}{\sqrt{2\pi (1-\tau^2)}}$ and introducing $\tilde{O}_1(z)=\frac{\sqrt{1-\tau^2}}{2\sqrt{2\pi}}(1-\frac{z^2}{(1+\tau)^2})$, we see that (\ref{eq:BulkScaling}) is exactly of the form (\ref{eq:BulkLimit}) for $\beta=1$.


When approaching the boundary $|z|=1+\tau$ of the eigenvalue support the typical diagonal overlap $\tilde{O}_1(z)$ tends to zero, and
in the appropriate scaling vicinity of the boundary it becomes parametrically weaker, as the variable $t$ in such a regime becomes of the order $\sqrt{N}$:
\begin{corollary} \label{th:edge}
 (edge scaling) Take a fixed $0\le \tau<1$ and parametrize $z$ and $q$ as $z=\sqrt{N}(1+\tau)+\delta_{\tau}\sqrt{1-\tau^2}$ and $q=\sigma \sqrt{N(1-\tau^2)}$. Then the limit ${\cal P}_{edge}(\delta_{\tau},\sigma)=\lim_{N\to\infty} \sqrt{N}{\cal P}_N(z,q)$ exists and
 is equal to
\begin{equation}
{\cal P}_{edge}(\delta_{\tau},\sigma)=\frac{1}{4\pi \sigma^2(1-\tau^2)}e^{-\frac{1}{4\sigma^2}+\frac{\delta_{\tau}}{\sigma}}\left[ e^{-2\delta_{\tau}^2}+\left(\frac{1}{\sigma}-2\delta_{\tau}\right) \int_{2\delta_{\tau}}^{\infty} e^{-\frac{u^2}{2}}du \right]. \label{eq:EdgeScaling}
\end{equation}
\end{corollary}
Note that this form is essentially the same as the one found for the real Ginibre case in \cite{F2018}.

Finally, the ultimate goal of our study is to investigate the {\it weak non-Hermiticity} regime occuring for $N\to \infty$ when $\tau$ approaches unity with the rate $O(N^{-1})$, so that the parameter $2N(1-\tau)=a^2$ is fixed. Such parameter therefore controls the deviation from the fully symmetric
limit. In this regime of `almost symmetric' matrices  already a finite fraction of eigenvalues of the order of $N$ are real, and asymptotically their mean density  is given by \cite{Efetov1997,FKS1998}
\begin{equation}
\left\langle\rho(z)\right\rangle=\rho_{sc}(z)\int_0^{1} e^{-\frac{1}{2}As^2} ds, \qquad  |z|<2, \label{eq:PWeakNonhden}
\end{equation}
where $\rho_{sc}(z)=\frac{1}{2\pi}\sqrt{4-z^2}$ is the standard Wigner semicircle density characterizing real symmetric GOE matrices,
and $A=(\pi \rho_{sc}(z)a)^2$. 

As anticipated, such a regime turns out to be not only `weakly non-Hermitian', but also `weakly non-normal'
as the typical value of the diagonal overlap $t=O_{ii}-1$ turns out to be of the order of unity in the bulk of the spectrum, namely
\begin{corollary} \label{th:weaknonh}
Let $|z|<2$ and $t\ge 0$ be fixed. Consider the limit ${\cal P}_{weak}(z,t)=\lim_{N\to\infty} N^{-1/2}{\cal P}_N(z\sqrt{N},t)$  with
$2N(1-\tau)=a^2$ fixed. Then
\begin{equation}
{\cal P}_{weak}(z,t)=\frac{A}{2}\rho_{sc}(z)\frac{e^{-\frac{A}{2t}}}{t^2}\int_0^1 e^{-\frac{1}{2}As^2}\left(1+A+\frac{A}{t}-As^2\right)s^2ds , \label{eq:PWeakNonh}
\end{equation}
where $A=(\pi \rho_{sc}(z)a)^2$ and $\rho_{sc}(z)=\frac{1}{2\pi}\sqrt{4-z^2}$.
\end{corollary}

\begin{remark} After integration by parts one can rewrite the above as
\begin{equation}
{\cal P}_{weak}(z,t)=\frac{A}{2}\rho_{sc}(z)\frac{e^{-\frac{A}{2t}}}{t^2}\left[\left(\frac{2}{A}-\frac{1}{t}\right) e^{-\frac{A}{2}}+\left(1+\frac{1}{t}-\frac{2}{A}\right)\int_0^1 e^{-\frac{1}{2}As^2} ds\right].
\end{equation}
From this form it is easy to check that $\int_0^{\infty} \cP_{weak}(z,t)dt $ agrees with the mean density (\ref{eq:PWeakNonhden}), as expected.
\end{remark}

We thus conclude that the characteristic tail ${\cal P}_{weak}(z,t)\sim t^{-2}$ is the most robust feature of the condition number density for real eigenvalues of asymmetric matrices, as it survives in the regime of arbitrary weak asymmetry as long as $a>0$,

\begin{figure}
\includegraphics[width=0.5\textwidth]{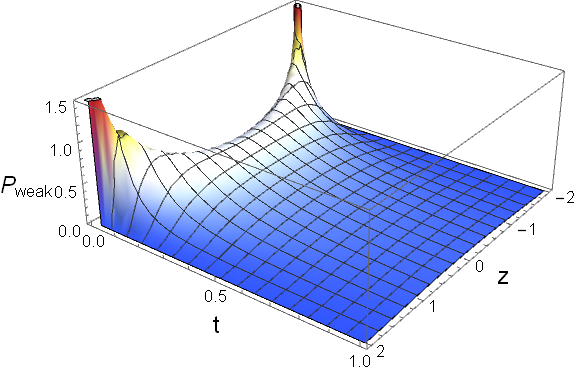} \includegraphics[width=0.5\textwidth]{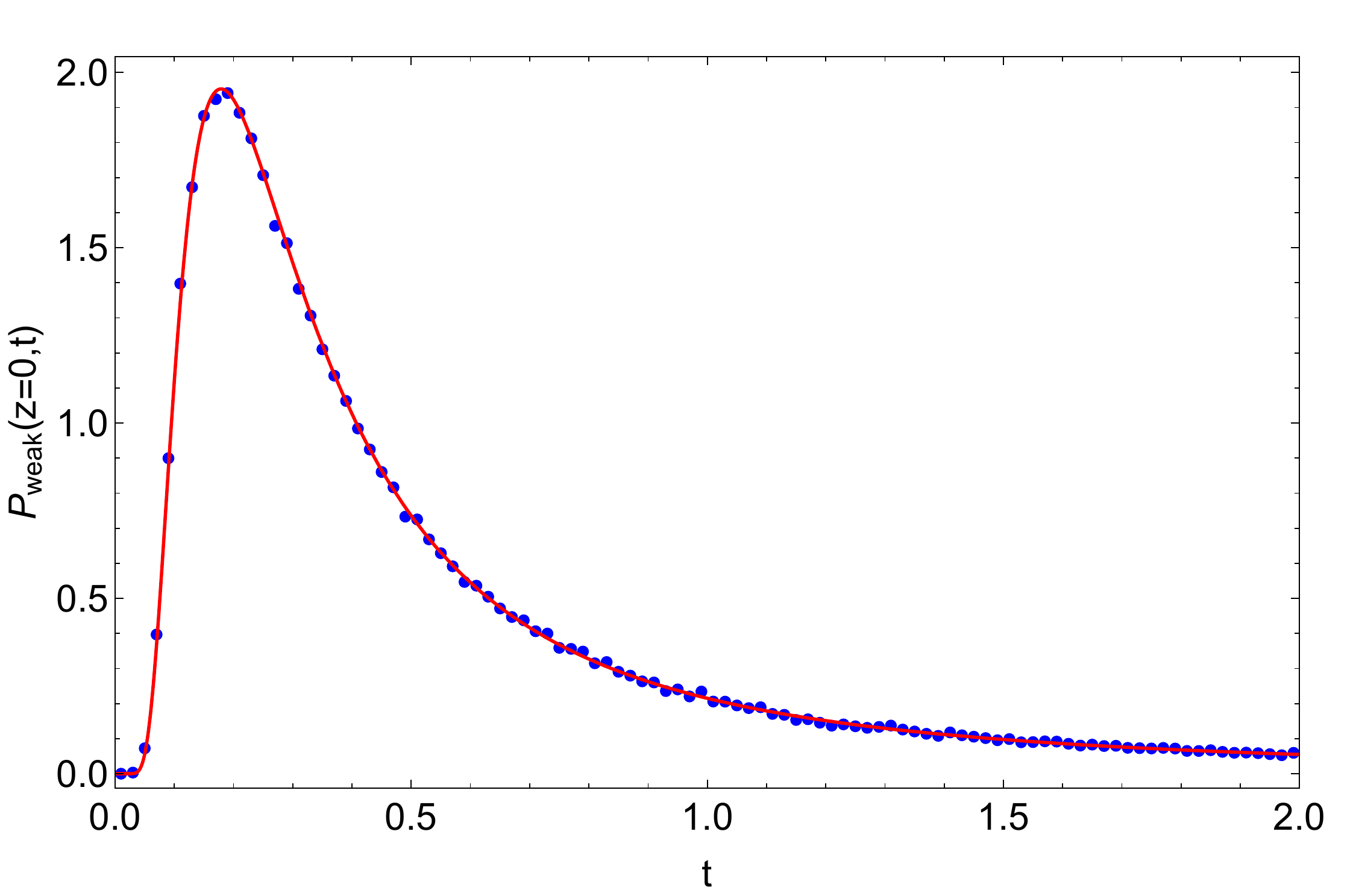}
\caption{
(left) 3D plot of $\mathcal{P}_{weak}(z,t)$. (right) Section of the plot at $z=0$ (red line) juxtaposed numerical diagonalization of $2\cdot 10^4$ matrices of size $N=500$.
}
\end{figure}

\textbf{Acknowledgments.} WT appreciates the support of Polish Ministry of Science and Higher Education  through the `Diamond Grant' 0225/DIA/2015/44 and the doctoral scholarship ETIUDA UMO-2018/28/T/ST1/00470 from National Science Center. WT is grateful to King's College London for warm hospitality during his stay. We are grateful to Janina Krzysiak for carefully reading this manuscript.

\section{Derivation of the main results}
We briefly outline an adaptation of the  method of evaluating the JDF in (\ref{eq:jointPdfDef}) following \cite{F2018} with necessary modifications.
\subsection{Partial Schur decomposition}
Let $\lambda$ be a real eigenvalue of an $N\times N$ real matrix $X_N$. Then it is well known, see e.g.\cite{EKS1994}, that the matrix $X_N$ can be decomposed as
\begin{equation}\label{decomp}
X_N=O\left(\begin{array}{cc}
\lambda & \bb{w}^T \\
0 & X_{N-1}
\end{array}\right)O^T=O\tilde{X}_NO^T,
\end{equation}
where $\bb{w}$ is a column vector with $N-1$ components and $X_{N-1}$ is a matrix of size $(N-1)\times (N-1)$. The matrix $O$ is the Householder reflection matrix, parameterized as $O=\idm_N - 2\bb{v}\bb{v}^T$, where $\idm_N$ is the $N\times N$ identity matrix and $\bb{v}$ is an $N$-dimensional vector of unit norm, with a positive first component.
Note that  although the left/right eigenvectors of $\tilde{X}_N$ corresponding to $\lambda$ are different from those of $X_{N}$, the inner products (hence, the eigenvalue condition numbers)  are the same. Parameterizing these eigenvectors as $\tilde{\bb{r}}_{\lambda}=(1,0,\ldots,0)^T$ and $\tilde{\bb{l}}_{\lambda}=(1,b_1,\ldots,b_{N-1})=(1,\mathbf{b}^T)$, we immediately obtain for the associated condition number $\kappa_{\lambda}^2=1+\bb{b}^T\bb{b}$. Demanding that
$\tilde{\bb{l}}_{\lambda}$ is the left eigenvector of $\tilde{X}_N$ leads us to the relation $\bb{b}=(\lambda-X_{N-1}^T)^{-1}\bb{w}$. As a consequence~\cite{F2018}
\begin{equation}
\kappa_{\lambda}^2=1+\bb{w}^T(\lambda-X_{N-1})^{-1}(\lambda-X_{N-1}^T)^{-1}\bb{w}.
\end{equation}
The Lebesgue measure on $X_{N}$ can be decomposed as $dX_{N}=|\det (\lambda-X_{N-1})| d\lambda d\bb{w} dX_{N-1} dS(\bb{v})$,
where $dS(\bb{v})$ denotes the uniform measure on the $(N-1)$-hemisphere~\cite{EKS1994} (see also supplementary material of~\cite{FK2016}).

It turns out to be more technically convenient to concentrate on evaluating a characteristic function $\cL(z,p)=\<\delta(z-\lambda) e^{-p\bb{b}^T\bb{b}}\>_{N}$ representing the Laplace transform of the JDF $\cP(z,t)$. Hereafter by $\<\ldots\>_{N}$ we denote the expected value with respect to the probability measure \eqref{eq:EllipticPDF} for matrices $X_N$ of size $N$.

\begin{lemma} \label{th:Laplace}
The characteristic function $\cL(z,p)$ can be represented in the form
\begin{equation}
\cL(z,p)=\frac{e^{-\frac{z^2}{2(1+\tau)}}}{2^{\frac{N}{2}}\Gmm{\frac{N}{2}}\sqrt{1+\tau}}\<\frac{\det(z-X)(z-X^T)}{\det{}^{1/2}\left[2p(1-\tau^2)+(z-X)(z-X^T)\right]}\>_{N-1}. \label{eq:Laplace}
\end{equation}
\end{lemma}
\begin{proof}
Substituting the decomposition (\ref{decomp}) together with the associated decomposition of the Lebesgue measure into the probability measure of the elliptic ensemble \eqref{eq:EllipticPDF}
one can easily see that the ensemble average  in \eqref{eq:Laplace} amounts to performing the following integral:
 \begin{multline}
 \cL(z,p)=C_{N}^{-1} e^{-\frac{z^2}{2(1+\tau)}}\int \exp\left[-\frac{1}{2(1-\tau^2)}\tr (X_{N-1}X_{N-1}^T-\tau X_{N-1}^2)\right] \times
 \\
 \exp\left[-\frac{1}{2(1-\tau^2)} \bb{w}^T\left(1+2p(1-\tau^2)(z-X_{N-1})^{-1}(z-X^T_{N-1})^{-1}\right)\bb{w}\right]|z-X_{N-1}|dX_{N-1}d\bb{w} dS(\bb{v}).
 \end{multline}
 The integral over $\bb{v}$ yields  the volume of the $(N-1)$-hemisphere $\int dS(\bb{v}) =\frac{\pi^{N/2}}{\Gmm{\frac{N}{2}}}$~\cite{EKS1994}. The integral over $\bb{w}$ is Gaussian and can be easily performed, giving the factor
 \begin{equation}
\frac{ \left[2\pi(1-\tau^2)\right]^{\frac{N-1}{2}}}{\det{}^{1/2}[1+2p(1-\tau^2)(z-X_{N-1})^{-1}(z-X_{N-1}^{T})^{-1}]}=\frac{\left[2\pi(1-\tau^2)\right]^{\frac{N-1}{2}} \det{}^{1/2}(z-X_{N-1})(z-X_{N-1}^T)}{\det{}^{1/2}\left[(2p(1-\tau^2)+(z-X_{N-1})(z-X_{N-1}^T)\right]}.
 \end{equation}
 Taking all the multiplicative numerical constants into account and the factor $|\det(z-X_{N-1})|$ from the Jacobian, we arrive at \eqref{eq:Laplace}.
\end{proof}

\subsection{Ratio of determinants}
The problem has been therefore reduced to the calculation of the expectation for the ratio of two random determinants
\begin{equation}
D_{N}:=\<\frac{ \det(z-X)(z-X^T)}{\det{}^{1/2}\left[(2p(1-\tau^2)+(z-X)(z-X^T)\right]}\>_{N},
\end{equation}
which is evaluated as
\begin{theorem} \label{th:RatioDet}
\begin{multline}
D_{N}=\frac{2^{-\frac{N}{2}}}{\sqrt{1+\tau}\Gmm{\frac{N}{2}}} \int_0^{\infty}\frac{dt}{t} e^{-pt(1-\tau)} e^{-\frac{z^2t}{2(1+\tau)(1+t)}} \left(\frac{t}{1+t}\right)^{1/2}\left(\frac{t}{1+\tau+t}\right)^{\frac{N-1}{2}}\times
\\
\left[\frac{P_{N-1}(1+\tau-2z^2)+2z(R_{N-1}+\tau(N-1)R_{N-2})}{1+t}+
\frac{P_{N-1}z^2}{(1+t)^2}+
\frac{\tau^2(1+\tau)^2(N^2-1)P_{N-2}}{(1+\tau+t)^2}+
\right.
\\
\left. 
\frac{(1+\tau)(1-\tau^2)(N-1)[(N-1)P_{N-2}-T_{N-2}]}{1+\tau+t}
-\frac{2\tau(1+\tau)(N-1)zR_{N-2}}{(1+t)(1+\tau+t)} \right],
\label{eq:RatioDeterminants}
\end{multline}
where $P_N$, $R_N$ and $T_N$ are defined in \eqref{eq:PNdef}-\eqref{eq:TNdef}.
\end{theorem}

\begin{remark} Theorem \ref{th:RatioDet}  immediately implies our main statement, Theorem \ref{th:Main}:
Indeed, by inserting \eqref{eq:RatioDeterminants} into \eqref{eq:Laplace} we see that $\cL(z,p)$ is already represented as a Laplace transform, but of the rescaled variable $t(1-\tau)$, and \eqref{eq:MainResult} follows. 
\end{remark}
\vspace{1cm}

The proof of Theorem \ref{th:RatioDet} proceeds via employing the supersymmetry approach to ratios of determinants.

\begin{proof}
Let $\chi, \rho, \theta, \eta$ denote $N$-component vectors in anticommuting (Grassmann) variables. This allows us to rewrite the determinant in the numerator as a standard Berezin Gaussian integral
\begin{equation}
\det (z-X)(z-X^T)=\int d\chi d\rho d\theta d\eta \exp\left[-\chi^T(z-X)\eta -\theta^T(z-X^T)\rho\right].
\end{equation}
The inverse square root of the determinant of a symmetric positive definite matrix $A$ can be represented as a standard Gaussian integral. 
Namely, introducing $N$-component real vectors $\bb{s}_1,\bb{s}_2$ we can write
\begin{multline}
\det{}^{-\frac{1}{2}}\left[2p(1-\tau^2)+(z-X)(z-X^T)\right]=
\\
\frac{1}{(2\pi)^{N}}\int d\bb{s}_1 d\bb{s}_2 \exp\left[-\frac{1}{2}(\bb{s}_1^T \bb{s}_2^T)\left(\begin{array}{cc}
u & i(z-X) \\
i(z-X^T) & u
\end{array}\right)
\left(\begin{array}{c}
\bb{s}_1 \\
\bb{s}_2
\end{array}\right)
\right],
\end{multline}
where we denoted $u^2=2p(1-\tau^2)$. This provides a representation of the right-hand side in \eqref{eq:RatioDeterminants} in the form
\begin{multline}
D_{N}=\frac{1}{(2\pi)^N}\int d\chi d\rho d\theta d\eta dS_1 dS_2 \exp\left[-z(\chi^T\eta+\theta^T\rho)-\frac{1}{2}(u\bb{s}_1^T\bb{s}_1+u \bb{s}_2^T\bb{s}_2+2i z \bb{s}_1^T\bb{s}_2)\right]\times
\\
 \< e^{\tr X(\theta\rho^T-\eta\chi^T+i\bb{s}_2\bb{s}_1^T)}\>_N. \label{eq:DNfirst}
\end{multline}
The identity $\<e^{-\tr XA}\>_N=e^{\frac{1}{2}\tr(AA^T+\tau A^2)}$ allows us to perform the ensemble average. This in turn produces terms that are quartic in Grassmann variables, which we further bilinearize by employing a few auxilliary Gaussian integrals, the step known as the Hubbard-Stratonovich transformation:
\begin{eqnarray}
e^{\theta^T\eta \rho^T\chi} &= & \frac{1}{\pi}\int_{\mathbb{C}} d^2a e^{-|a|^2+a\theta^T\eta+\bar{a}\rho^T\chi}, \qquad
e^{\tau\chi^T\theta\rho^T\eta}  =   \frac{1}{\pi} \int_{\mathbb{C}} d^2 b e^{-|b|^2+\sqrt{\tau}b\chi^T\theta+\sqrt{\tau}\bar{b}\rho^T\eta},
\\
e^{-\frac{\tau}{2}(\rho^T\theta)^2} & = & \frac{1}{\sqrt{2\pi}} \int_{\mathbb{R}} e^{-\frac{c^2}{2}+ic\sqrt{\tau}\rho^T\theta}dc, \qquad
e^{-\frac{\tau}{2}(\chi^T\eta)^2}  =  \frac{1}{\sqrt{2\pi}} \int_{\mathbb{R}} e^{-\frac{f^2}{2}-if\sqrt{\tau}\chi^T\eta}df,
\end{eqnarray}
where we use the notation $d^2z=dx\, dy$ for $z=x+iy$.

Applying these transformations converts all integrations over anticommuting variables into a Gaussian Berezin integral which we can write as  $\int  d\chi d\rho d\theta d\eta e^{-\frac{1}{2}\xi^T M \xi}$, where $\xi^T=(\chi^T\eta^T\theta^T \rho^T)$ and the antisymmetric matrix $M$ is given by
\begin{equation} \label{eq:Mdef}
M=\left(\begin{array}{cccc}
0 & g-iA & -b\sqrt{\tau} & \bar{a} \\
-g+iA^T & 0 & a & \bar{b}\sqrt{\tau} \\
b\sqrt{\tau} & -a & 0 & h-i A^T \\
-\bar{a} & -\bar{b}\sqrt{\tau} & -h+iA & 0
\end{array}\right).
\end{equation}
Here  we denoted $g=z+if\sqrt{\tau}$, $h=z+ic\sqrt{\tau}$ for brevity, and introduced the rank-two matrix $A=\bb{s}_1\bb{s}_2^T+\tau \bb{s}_2 \bb{s}_1^T$.

The Berezin Gaussian integration yields a Pfaffian of the matrix $M$. The Pfaffian is evaluated using the fact that through the Schur decomposition $A$ is equivalent to $\tilde{A}\oplus 0_{N-2}$, where $0_{N-2}$ denotes the matrix of size $N-2$ filled with zeros, and
\begin{equation}
\tilde{A} = \begin{pmatrix}
\lambda_1 & t \\
0 & \lambda_2
\end{pmatrix}.
\end{equation}
As a consequence, $\textup{Pf}(M) = \textup{Pf}(F)^{N-2} \textup{Pf}(G)$, where the matrices of size $4$ and $8$ read explicitly

\begin{equation}
F = \left(\begin{array}{cccc}
0 & g & -b\sqrt{\tau} & \bar{a} \\
-g & 0 & a & \bar{b}\sqrt{\tau} \\
b\sqrt{\tau} & -a & 0 & h \\
-\bar{a} & -\bar{b}\sqrt{\tau} & -h & 0
\end{array}\right),
\qquad
G=F\otimes \idm_2 + \left(\begin{array}{cccc}
0 & - i\tilde{A} & 0 & 0 \\
i \tilde{A}^T & 0 & 0 & 0 \\
0 & 0 & 0 & i \tilde{A}^T \\
0 & 0 & -i \tilde{A} & 0 
\end{array}\right).
\end{equation}
The Pfaffian reads
\begin{multline} \label{eq:pfaffianM}
\mathrm{Pf}(M)=(|a|^2+\tau |b|^2+gh)^{N-2}\left[ (|a|^2+\tau |b|^2+gh)^{2}-(|a|^2+\tau |b|^2+gh)i(g+h)\tr A- \right.
\\
\left. |a|^2 \tr AA^T-\tau|b|^2\tr A^2-(g^2+h^2)\det A- gh(\tr A)^2-i (g+h) \det A \,  \tr A+\det{}^2A\right].
\end{multline}
We notice that
\begin{displaymath} 
\begin{array}{cclccl}
\tr A &=& (1+\tau) \bb{s}_1^T\bb{s}_2,&\quad \tr A^2 &=& (1+\tau^2)(\bb{s}_1^T\bb{s}_2)^2+2\tau (\bb{s}_1^T\bb{s}_1)(\bb{s}_2^T\bb{s}_2), \\
\tr AA^T &=& (1+\tau^2)(\bb{s}_1^T\bb{s}_1)(\bb{s}_2^T\bb{s}_2) + 2\tau (\bb{s}_1^T\bb{s}_2)^2, & \quad \det A &=& \tau (\bb{s}_1^T\bb{s}_2)^2-\tau(\bb{s}_1^T\bb{s}_1)(\bb{s}_2^T\bb{s}_2).
\end{array}
\end{displaymath}
 Therefore, the resulting expression depends on the vectors $\bb{s}_1$ and $\bb{s}_2$ only via their scalar products, so it is convenient to parameterize integrals by the entries of the associated Gram matrix~\cite[Section 2]{F2002}\cite[Theorem 1a in Appendix D]{FS2002}
\begin{eqnarray}
\hat{Q}=\begin{pmatrix}
Q_1 & Q \\
Q & Q_2
\end{pmatrix},
\qquad
\hat{Q}_{ij}=(\bb{s}^T_i\bb{s}_j), \qquad i,j=1,2. \label{eq:changeQ}
\end{eqnarray}

The Jacobian of this change of variables is $(\det \hat{Q})^{\frac{N-3}{2}}$, while the integration over redundant angular variables yields the factor $C_{N,2}^{(o)} =\frac{2^{N-2}\pi^{N-1}}{\Gmm{N-1}}$~\cite{FS2002}. The range of integration is restricted by the non-negativity conditions $Q_1\ge 0,\,Q_2\ge 0, \det\hat{Q}=Q_1Q_2-Q^2\ge 0$. Following~\cite{F2018} it is convenient to change variables into $r=(\det \hat{Q})^{1/2}$ and parameterize the integration region by $Q_2=\frac{r^2+Q^2}{Q_1}$. The change of measure reads $dQ_1 dQ_2 dQ=2\frac{dQ_1}{Q_1}r dr dQ$. After rescaling $Q_1\to uQ_1$, we have
\begin{multline}
D_{N}=\frac{1}{4\pi^4\Gmm{N-1}}\int_{\mathbb{C}} d^2a\int_{\mathbb{C}}d^2b\, e^{-|a|^2-|b|^2} \int_{\mathbb{R}^2}dc\,df\, e^{-\frac{c^2}{2}-\frac{f^2}{2}}
\\
\times\,\int_{-\infty}^{\infty}dQ\int_{0}^{\infty}\frac{dQ_1}{Q_1} \int_0^{\infty}r^{N-2}dr  \,  e^{-\frac{1}{2}\left(u^2Q_1+\frac{r^2+Q^2}{Q_1}+2izQ+r^2+Q^2(1+\tau)\right)} \rm{Pf}(M). \label{eq:Int1}
\end{multline}
After the change of variables \eqref{eq:changeQ}, $\mathrm{Pf} (M)$ can be expressed as
\begin{multline}
(|a|^2+\tau|b|^2+gh)^{N-2}\left[(|a|^2+\tau|b|^2+gh)^2-(|a|^2+\tau|b|^2+gh)(iQ(1+\tau)(g+h)+Q^2(1+\tau)^2)-
\right.
\\
\left.
-|a|^2(1+\tau^2)r^2-2\tau^2r^2|b|^2+\tau r^2(g^2+h^2)-i\tau(g+h)Q(1+\tau)r^2+\tau^2 r^4\right].
\end{multline}
The integrals over $a,b,c,f$ are performed in the following way. Let us denote
\begin{equation}
P_N=\frac{1}{2\pi^3}\int d^2a d^2b dc df e^{-|a|^2-|b|^2-\frac{c^2}{2}-\frac{f^2}{2}} (|a|^2+\tau |b|^2+gh)^{N}.
\end{equation}
Expanding the expression in the bracket and using the binomial theorem twice, we obtain
\begin{equation}
P_{N}=N!\sum_{k=0}^{N}\tau^k\sum_{m=0}^{k}\frac{1}{m!}\He_{m}^2\left(\frac{z}{\sqrt{\tau}}\right),
\end{equation}
where $\He_{m}(x)=(2\pi)^{-1/2}\int_{-\infty}^{\infty} e^{-\frac{y^2}{2}}(x+iy)^m dy$ are the monic Hermite polynomials. The internal sum can be performed via the Cristoffel-Darboux formula, finally yielding
\begin{equation}
P_N=N!\sum_{k=0}^{N}\frac{\tau^k}{k!}\left( (k+1)\He_k^2\left(\frac{z}{\sqrt{\tau}}\right)-k \He_{k-1}\left(\frac{z}{\sqrt{\tau}}\right)\He_{k+1}\left(\frac{z}{\sqrt{\tau}}\right) \right). \label{eq:Pdef2}
\end{equation}
Note that $P_{N}$ can be interpreted as the expectation of the squared characteristic polynomial $\<\det(z-X)(z-X^T)\>_{N}$ and in this capacity has been already studied for the Real Elliptic ensemble~\cite{APS2008}. All other integrals over $a, b, c, f$ in \eqref{eq:Int1} are performed in a similar way. After exploiting the three term recurrence for Hermite polynomials $\He_{N+1}(x)=x\He_{N}(x)-N\He_{N-1}(x)$,
the integrals are evaluated to
\begin{multline}
P_{N}-P_{N-1}Q^2(1+\tau)^2+r^4\tau^2P_{N-2}-r^2[(N-1)(1+\tau^2)+4\tau^2]P_{N-2}-2iQ(1+\tau)R_{N-1}+
\\
2r^2\tau(z-iQ(1+\tau))R_{N-2}+(1-\tau^2)r^2T_{N-2}, \label{eq:abcgIntegrated}
\end{multline}
where $R_N$ and $T_N$ are defined by \eqref{eq:RNdef} and \eqref{eq:TNdef}. Note also that $R_{N}(z)=\frac{1}{2(N+1)}\frac{dP_{N+1}(z)}{dz}$.
It is convenient to exploit the structure of \eqref{eq:abcgIntegrated} and exponent in \eqref{eq:Int1} and rescale further $Q\to\frac{Q}{1+\tau}$ and, similarly, $Q_1\to\frac{Q_1}{1+\tau}$. Recalling that $u^2=2p(1-\tau^2)$, one then arrives at
\begin{multline}
D_{N}=\frac{1}{2\pi(1+\tau)\Gmm{N-1}}\int_0^{\infty}\frac{dQ_1}{Q_1}e^{-pQ_1(1-\tau)}\int_{\mathbb{R}}dQ e^{-\frac{1}{2(1+\tau)}\left(Q^2\frac{1+Q_1}{Q_1}+2izQ\right)}
\\
\int_0^{\infty}r^{N-2} dr  e^{-\frac{r^2}{2}\frac{Q_1+1+\tau}{Q_1}}   \left[ P_N-P_{N-1}Q^2+r^4\tau^2 P_{N-2}-r^2[(N-1)(1+\tau^2)+4\tau^2]P_{N-2} - \right.
\\
\left. 2iQR_{N-1}+2r^2\tau(z-iQ)R_{N-2}+
(1-\tau^2)r^2T_{N-2}\right].
\end{multline}
The integrations over $Q$ and $r$ are elementary. The integral over $Q$ is Gaussian, while the one over $r$ is of the type $\int_0^{\infty}r^{N-2} e^{-ar^2/2}dr=\frac{1}{2}\left(\frac{2}{a}\right)^{\frac{N-1}{2}}\Gmm{\frac{N-1}{2}}$.  
The remaining integral over $Q_1$ formally looks logarithmically divergent. To see the cancellation of the divergent part, one should exploit a non-trivial identity
\begin{equation}
P_N-P_{N-1}(1+\tau-z^2)-(N-1)[2\tau^2+N-1]P_{N-2}-2zR_{N-1}+(1-\tau^2)(N-1)T_{N-2}=0, \label{eq:Identity}
\end{equation}
which is verified in Appendix~\ref{sec:App1}.  After elementary but tedious algebraic manipulations with the help of \textit{Mathematica} in rearranging terms, we finally obtain
\begin{multline}
D_N=\frac{2^{-\frac{N}{2}}}{\sqrt{1+\tau}\Gmm{\frac{N}{2}}}\int_0^{\infty}\frac{dQ_1}{Q_1} e^{-pQ_1(1-\tau)} e^{-\frac{z^2Q_1}{2(1+\tau)(1+Q_1)}}\left(\frac{Q_1}{1+Q_1}\right)^{\frac{1}{2}} \left(\frac{Q_1}{1+\tau+Q_1}\right)^{\frac{N-1}{2}}\times
\\
\left[ \frac{P_{N-1}(1+\tau-2z^2)+2z(R_{N-1}+\tau(N-1)R_{N-2})}{1+Q_1}+\frac{P_{N-1}z^2}{(1+Q_1)^2}+\frac{\tau^2(1+\tau)^2(N^2-1)P_{N-2}}{(1+\tau+Q_1)^2}+  \right.
\\
\left.
\frac{(1+\tau)(1-\tau^2)(N-1)[(N-1)P_{N-2}-T_{N-2}]}{1+\tau+Q_1}-\frac{2\tau(1+\tau)(N-1)zR_{N-2}}{(1+Q_1)(1+\tau+Q_1)}
\right]. \label{eq:Dn}
\end{multline}
\end{proof}

\subsection{Asymptotic analysis}
As $P_N$, $R_N$ and $T_N$ are the building blocks of the determinant, we consider here their large-$N$ asymptotics.
First, we find convenient integral representations, which should allow the use of the Laplace method. For this we
start from \eqref{eq:Pdef2} and, using the integral representation for Hermite polynomials in (\ref{eq:Hermite}) with both signs, we obtain
\begin{equation}
P_N=\frac{N!}{2\pi\tau}e^{\frac{z^2}{\tau}}\sum_{k=0}^{N}\frac{1}{k!} \int_{\mathbb{R}^2} dt ds e^{-\frac{t^2+s^2}{2\tau}-\frac{iz}{\tau}(t-s)} [(k+1)t^ks^k+k t^{k+1} s^{k-1}].
\end{equation}
The sum is evaluated using $\sum_{k=0}^{N}\frac{x^k}{k!}=e^{x}\frac{\Gmm{N+1,x}}{\Gmm{N+1}}$, where $\Gmm{N+1,x}=\int_x^{\infty}u^{N}e^{-u}du$. This yields
\begin{lemma}
\begin{align}
P_{N}(z)=\frac{N!}{2\pi\tau} e^{\frac{z^2}{\tau}}\int_{\mathbb{R}^2} dt ds e^{-\frac{t^2+s^2}{2\tau}-\frac{iz}{\tau}(t-s)+ts}\left(\frac{\Gmm{N+1,ts}}{N!}+t(s+t)\frac{\Gmm{N,ts}}{(N-1)!}\right). \label{eq:PnIntegral}
\end{align}
\end{lemma}
An analogous procedure applied to $T_N$ gives
\begin{lemma}
\begin{equation}
T_{N}=\frac{N!}{2\pi \tau}e^{\frac{z^2}{\tau}}\int_{\mathbb{R}^2}dt ds e^{-\frac{t^2+s^2}{2\tau}-\frac{iz}{\tau}(t-s)+ts}\left((t^2+2ts)\frac{\Gmm{N,ts}}{(N-1)!}+t^2s(t+s)\frac{\Gmm{N-1,ts}}{(N-2)!}\right). \label{eq:TnIntegral}
\end{equation}
\end{lemma}

\subsubsection{Bulk scaling}
Let us give the proof of {\bf Corollary \ref{th:BulkScaling}}.
\begin{proof}
After rescaling $z\to z\sqrt{N}$, $t\to t\sqrt{N}$ and $s\to s\sqrt{N}$, and then changing the integration variables
$(t,s)\to(p,q)$ as $(t+s)/\sqrt{2}=p$ and $(t-s)/\sqrt{2}=q$ the equation
\eqref{eq:PnIntegral} takes the following form:
\begin{equation}
P_N(z\sqrt{N})=\frac{ N! N}{2\pi \tau}e^{N\frac{z^2}{\tau}}\int_{\mathbb{R}}dp\,e^{-N\frac{p^2}{2}\left(\frac{1}{\tau}-1\right)}
\int_{\mathbb{R}}dq e^{-N\left(\frac{q^2}{2}\left(\frac{1}{\tau}+1\right)+\frac{iz\sqrt{2}}{\tau}q\right)} \label{asyintrep1}
\end{equation}
\[
\times \left( \theta_{N}\left(\frac{p^2-q^2}{2}\right)+Np^2\theta_{N-1}\left(\frac{p^2-q^2}{2}\frac{N}{N-1}\right)\right),
\]
where we denoted $\theta_{N}(x)=\frac{\Gmm{N+1,Nx}}{\Gmm{N+1}}$. 
Note that for any $N$ this function is bounded: $\theta_{N}(x)\le 1$, and in the limit $N\to \infty$ for a fixed real $x$ we have $\theta_{N}(x)\to \theta_{\infty}(x)$, where
$ \theta_{\infty}(x)=1$ for $x<1$ and $0$ otherwise.

For $N\gg 1$ the integral over $p$ can be most straightforwardly evaluated by the Laplace method, yielding that the leading contribution to  $P_N(z\sqrt{N})$ can be written as
  \begin{equation}
P_N(z\sqrt{N})\sim \frac{ N! \sqrt{N}}{\sqrt{2\pi \tau(1-\tau)}}\,e^{N\frac{z^2}{\tau}}
\int_{\mathbb{R}}dq e^{-N\left(\frac{q^2}{2}\left(\frac{1}{\tau}+1\right)+\frac{iz\sqrt{2}}{\tau}q\right)}
\left( \theta_{N}\left(\frac{-q^2}{2}\right)+\frac{\tau}{1-\tau}\theta_{N-1}\left(\frac{-q^2}{2}\right)\right). \label{asyintrep1a}
\end{equation}
 For large $N$ the $q$-integral above  can be performed (for a fixed, $N-$independent real value of $z$)  by the standard saddle point method, with the saddle point position given by $q=q_*:=-\frac{iz\sqrt{2}}{1+\tau}$ yielding the required asymptotic formula:
 \begin{equation}
P_{N}(z\sqrt{N})\sim \frac{N!}{\sqrt{1-\tau^2}(1-\tau)} e^{\frac{Nz^2}{1+\tau}} \theta_{\infty}\left(z^2/(1+\tau)^2\right). \label{eq:PnBulkAsympt}
\end{equation}
The same type of reasoning applied to~\eqref{eq:TnIntegral} gives
\begin{equation}
T_{N}(z\sqrt{N}) \sim \frac{N!}{(1-\tau^2)^{3/2}}\frac{Nz^2}{1+\tau} e^{\frac{Nz^2}{1+\tau}}\theta_{\infty}\left(z^2/(1+\tau)^2\right) .
\end{equation}
Finally, the asymptotics
\begin{equation}
R_{N}(z\sqrt{N})\sim \frac{N!z\sqrt{N}}{(1-\tau^2)^{3/2}} e^{\frac{Nz^2}{1+\tau}}\theta_{\infty}\left(z^2/(1+\tau)^2\right)
\end{equation}
is obtained from \eqref{eq:PnBulkAsympt} using the fact that $R_{N}(z)=\frac{1}{2(N+1)}\frac{dP_{N+1}(z)}{dz}$.

Upon inserting this asymptotics into \eqref{eq:MainResult} and rescaling $q\to Nq$ it is clear that only the second to last term in the square bracket provides the leading order contribution, which happens to be
\begin{equation}
\frac{(1+\tau)(1-\tau^2)(N-1)[(N-1)P_{N-2}-T_{N-2}]}{1+\tau+Nq}\sim \frac{(N-1)![(1+\tau)^2-z^2]}{q\sqrt{1-\tau^2}} e^{\frac{Nz^2}{1+\tau}}. \label{eq:AuxAympt}
\end{equation}
As a consequence, the joint pdf reads
\begin{equation}
N \mathcal{P}(z\sqrt{N},Nq)=\frac{\sqrt{\frac{1+\tau}{1-\tau}}(1-\frac{z^2}{(1+\tau)^2})}{2\sqrt{2\pi}q^2}e^{-\frac{1+\tau}{2q}\left(1-\frac{z^2}{(1+\tau)^2}\right)}.
\end{equation}
Changing variables to $t=\frac{q}{1-\tau}$, one immediately recovers Corollary \ref{th:BulkScaling}.
\end{proof}

\subsubsection{Edge scaling}
When $z$ is tuned to values parametrically close to $z=\pm (1+\tau)$ where the step-function argument in
equations (\ref{eq:PnBulkAsympt})-(\ref{eq:AuxAympt}) is close to unity with a distance $O(N^{-1/2})$, the correponding asymptotics need to be evaluated
with higher accuracy. Such a regime is known as the {\it edge scaling}, which features in {\bf Corollary \ref{th:edge}}
which we now prove.

\begin{proof}
In the proof we choose the vicinity of $z=1+\tau$. Correspondingly, in (\ref{asyintrep1}) we now scale $z=1+\tau+\frac{w}{\sqrt{N}}$, where $w$ is of order 1.  The transition from (\ref{asyintrep1}) to (\ref{asyintrep1a}) remains the same as before.
Now we use the integral representation of the incomplete gamma function $\Gmm{N,x}=x^N\int_1^{\infty} u^{N-1} e^{-ux}du$, which
helps to rewrite the integral (\ref{asyintrep1a}) in the form
\begin{multline}
P_N(z)\sim \frac{N! N^{N+1}}{4\pi \tau \Gmm{N}} e^{N\frac{z^2}{\tau}} \int_{\mathbb{R}} dq\, e^{-\frac{N(1+\tau)}{2\tau} (q^2+i2\sqrt{2}q)-\frac{iqw\sqrt{2N}}{\tau}} \int_1^{\infty}du\, e^{-\frac{Nu}{2}(-q^2)+N\ln\left[\frac{u}{2}(-q^2)\right]}\times 
\\
\left(\frac{-q^2}{2}+\frac{\tau}{1-\tau}\frac{1}{u}\right).\label{asyintrep2}
\end{multline}

An inspection shows that whereas the $q-$integration is dominated by the contribution from the saddle point  $q=-\sqrt{2}i$, the last $u$-integral is dominated by the vicinity of $u=1$ of the width $O(N^{-1/2})$. Parameterizing in such a vicinity
 $u=1+\frac{v}{\sqrt{N}}$ one then arrives at the leading-order asymptotics
\begin{equation}
P_N\sim \frac{N! N^{N-1/2}}{(1-\tau)^2\Gmm{N}} e^{\frac{w^2}{1+\tau}+2w\sqrt{N}+N\tau} \int_0^{\infty} e^{-\frac{1+\tau}{2(1-\tau)}\left(v+\frac{2w}{1+\tau}\right)^2}dv.
\end{equation}
After the change of variables $u=\sqrt{\frac{1+\tau}{1-\tau}}(v+\frac{2w}{1+\tau})$ and the use of Stirling's approximation $\Gmm{N+1}\sim \sqrt{2\pi}N^{N+\frac{1}{2}} e^{-N}$, we obtain
\begin{equation}
P_{N}=\frac{N! e^{\frac{z^2}{1+\tau}}}{(1-\tau)\sqrt{2\pi(1-\tau^2)}} \int\limits_{\frac{2w}{\sqrt{1-\tau^2}}}^{\infty} e^{-\frac{u^2}{2}}du.
\end{equation}
The last integral is related to the complementary error function $\mathrm{erfc}(x)=\frac{2}{\sqrt{\pi}}\int_x^{\infty} e^{-t^2}dt$. Using $R_N(z)=\frac{1}{2(N+1)}\frac{dP_{N+1}(z)}{dz}$ we obtain that in such a regime asymptotically $R_{N}\sim \sqrt{N}P_{N}$.
From the asymptotics of \eqref{eq:AuxAympt} one expects that the leading order contributions from $(N+1)P_N$ and $T_N$ cancel. Therefore, one needs to work with the appropriate integral representation. Combining \eqref{eq:PnIntegral} and \eqref{eq:TnIntegral}  and following the analogous reasoning as above we obtain
\begin{equation}
(N+1)P_N-T_N=\frac{N!\sqrt{N}}{\sqrt{2\pi}(1-\tau^2)} e^{\frac{z^2}{1+\tau}}\left[e^{-\frac{2w^2}{1-\tau^2}}-\frac{2w}{\sqrt{1-\tau^2}} \int_{\frac{2w}{\sqrt{1-\tau^2}}}^{\infty} e^{-\frac{u^2}{2}}du\right],
\end{equation}
which is of the same order as $\sqrt{N}P_N$, as expected. To get the correct asymptotics at the edge, we rescale $q\to q\sqrt{N}$ in \eqref{eq:MainResult}.  It is now clear that the first term in the square bracket in \eqref{eq:MainResult} is subleading and the  contribution of other terms is of the same order. The asymptotics of $\left(\frac{q\sqrt{N}}{1+\tau+q\sqrt{N}}\right)^{\frac{N}{2}-1}$ is calculated as
\begin{equation}
e^{-\frac{N}{2}\ln\left(1+\frac{1+\tau}{q\sqrt{N}}\right)}\sim e^{-\frac{(1+\tau)\sqrt{N}}{2q}+\frac{(1+\tau)^2}{4q^2}},
\end{equation}
and we obtain

\begin{equation}
\mathcal{P}\sim \frac{1}{4\pi q^2 \sqrt{N}}e^{\frac{w}{q}-\frac{1-\tau^2}{4q^2}}\left[ e^{-\frac{2w^2}{1-\tau^2}}+\left(\frac{\sqrt{1-\tau^2}}{q}-\frac{2w}{\sqrt{1-\tau^2}}\right)\int_{\frac{2w}{1-\tau^2}} e^{-\frac{u^2}{2}} du\right].
\end{equation}
After denoting $w=\delta_{\tau}\sqrt{1-\tau^2}$ and $q=\sigma\sqrt{1-\tau^2}$, the statement of Corollary \ref{th:edge} follows.

\end{proof}

\subsubsection{Weak non-Hermiticity}
Our final goal is to provide the proof of  {\bf Corollary \ref{th:weaknonh}}.

\begin{proof}
We start again from (\ref{asyintrep1})  keeping $z$ fixed and $N-$independent like before in the bulk case,
but for the weak non-Hermiticity regime we replace $\tau=1-\frac{a^2}{2N}$.  
It is then immediately obvious that the $p$-integral is no longer 
dominated by the small vicinity $p\sim N^{-1/2}$, but rather by the integration range of order unity. 
Then a quick inspection shows that for extracting the leading asymptotics 
 in the large $N$ limit one can effectively replace (\ref{asyintrep1}) by 
\begin{equation}
P_N(z\sqrt{N})\sim\frac{ N! N^2}{2\pi \tau}e^{N\frac{z^2}{\tau}}\int_{\mathbb{R}}dp\,p^2 e^{-\frac{p^2a^2}{4}}
\int_{\mathbb{R}}dq e^{-N\left(q^2+iz\sqrt{2}q\right)} \theta_{\infty}\left(\frac{p^2-q^2}{2}\right).
\end{equation}
Performing the integral over $q$ by the saddle point method, we see that the range of integration over $p$ is given by 
$|p|<\frac{\sqrt{4-z^2}}{\sqrt{2}}$. 
After a few simple changes of variables and straightforward manipulations one arrives at
\begin{equation}
P_N(z\sqrt{N})\sim \frac{N! N^{3/2}(4-z^2)^{3/2}}{2\sqrt{2\pi}} e^{\frac{Nz^2}{2}} \int_0^{1} e^{-\frac{s^2a^2}{2}\left(1-\frac{z^2}{4}\right)} s^2 ds. \label{eq:PnWeak}
\end{equation}
The asymptotic behavior of $R_N$ simply follows from the relation $R_{N}(z)=\frac{1}{2(N+1)}\frac{d P_{N+1}(z)}{dz}$ and is related to the asymptotics of $P_N$ as $R_N=\frac{z\sqrt{N}}{2}P_N$.  Asymptotics of $T_N$ analogously follows from its integral representation and  reads
\begin{align}
T_{N}(z\sqrt{N})=\frac{N!N^{5/2}(4-z^2)^{3/2}}{2\sqrt{2\pi}} e^{\frac{Nz^2}{2}}\int_0^1 e^{-\frac{s^2a^2}{2}\left(1-\frac{z^2}{4}\right)} \frac{4s^4-z^2s^4+z^2s^2}{4}ds. \label{eq:TnWeak}
\end{align}
Note that in \eqref{eq:MainResult} we used the rescaled quantity $q=(1-\tau)t$. Therefore, for the correct asymptotics, we need to rescale $q\to \frac{2N}{a^2}t$. This shows that all terms in the square bracket are of the same order. Direct use of the asymptotic forms \eqref{eq:PnWeak} and \eqref{eq:TnWeak} leads to \eqref{eq:PWeakNonh}.
\end{proof}


\begin{appendices}

\section{Density of real eigenvalues for moderate matrix size} \label{sec:AppDens}

For $N=2$ the joint pdf~\eqref{eq:MainResult} reads
\begin{equation}
{\cal P}_{N=2}(z,q)=\frac{1}{2\sqrt{2\pi}(1+\tau)} \frac{e^{-\frac{z^2}{2(1+\tau)}(1+\frac{q}{1+q})}}{\sqrt{q(1+q)}} \left(\frac{z^2}{(1+q)^2}+\frac{1+\tau}{1+q}\right).
\end{equation}
The substitution $t^2=\frac{q}{q+1}$ allows one to calculate the integral. After integration by parts, we obtain
\begin{equation}
\int_{0}^{\infty} {\cal P}_{N=2}(z,q)dq=\frac{e^{-\frac{z^2}{1+\tau}}}{\sqrt{2\pi}}+\frac{e^{-\frac{z^2}{2(1+\tau)}}}{\sqrt{2\pi}}\frac{z}{1+\tau}\int_0^z e^{-\frac{u^2}{2(1+\tau)}}du,
\end{equation}
which agrees with \eqref{eq:RealDens1}-\eqref{eq:RealDens2} when we substitute $N=2$. This way, with the help of \textit{Mathematica} software, we were also able to perform integration for $N=3,4$. For $N=4$ we again see agreement with the Forrester-Nagao result~\cite{ForNag2008}, while for $N=3$ we compared the results of integration with the numerical diagonalizations of random matrices, see Fig.~\ref{Fig:2}. For moderate matrix sizes, where the symbolic calculations were not possible, we numerically integrated~\eqref{eq:MainResult} and compared with numerical diagonalization, observing good agreement, see Fig.~\ref{Fig:2}.

\begin{figure}
\includegraphics[width=0.5\textwidth]{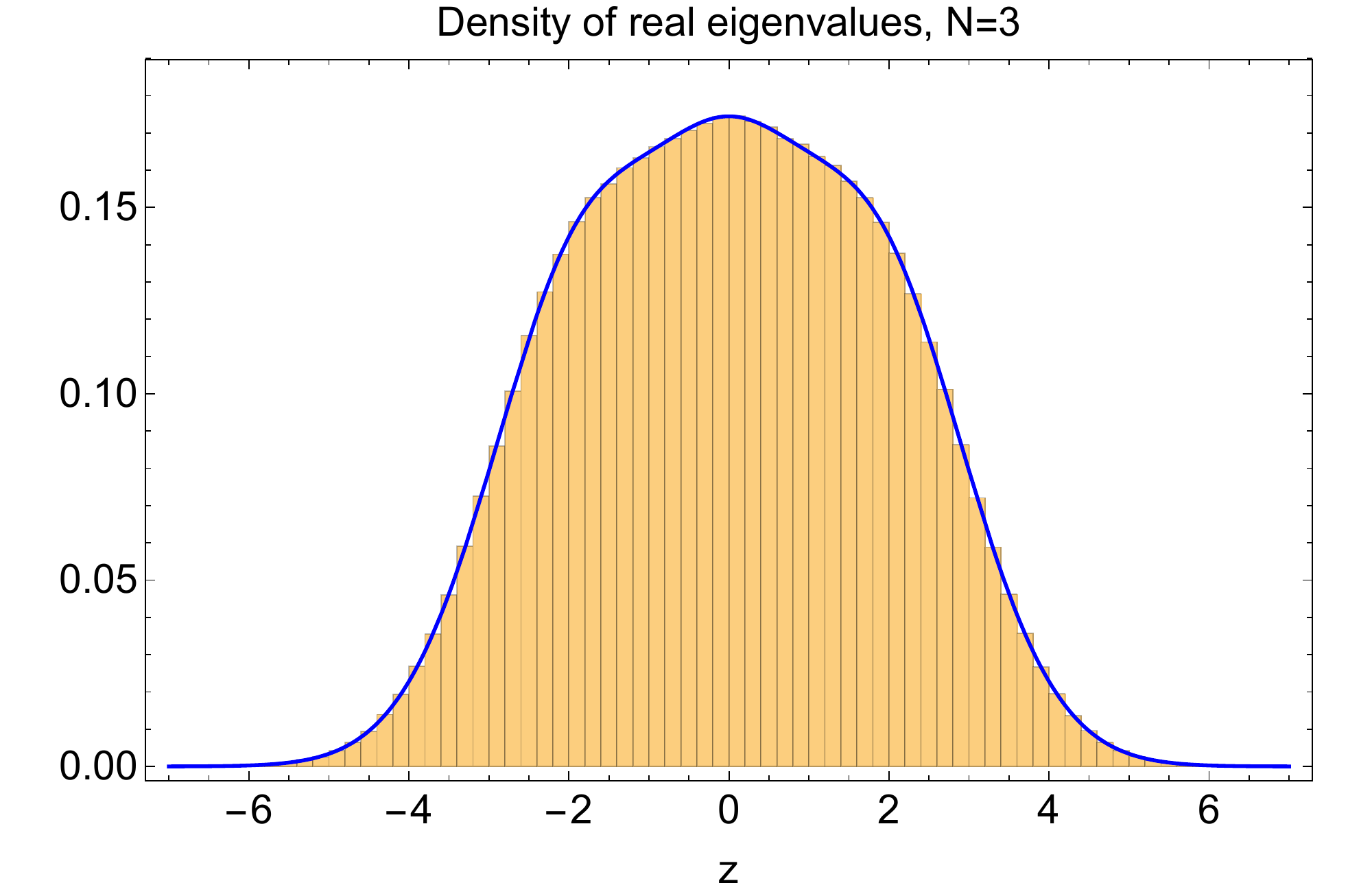}\includegraphics[width=0.5\textwidth]{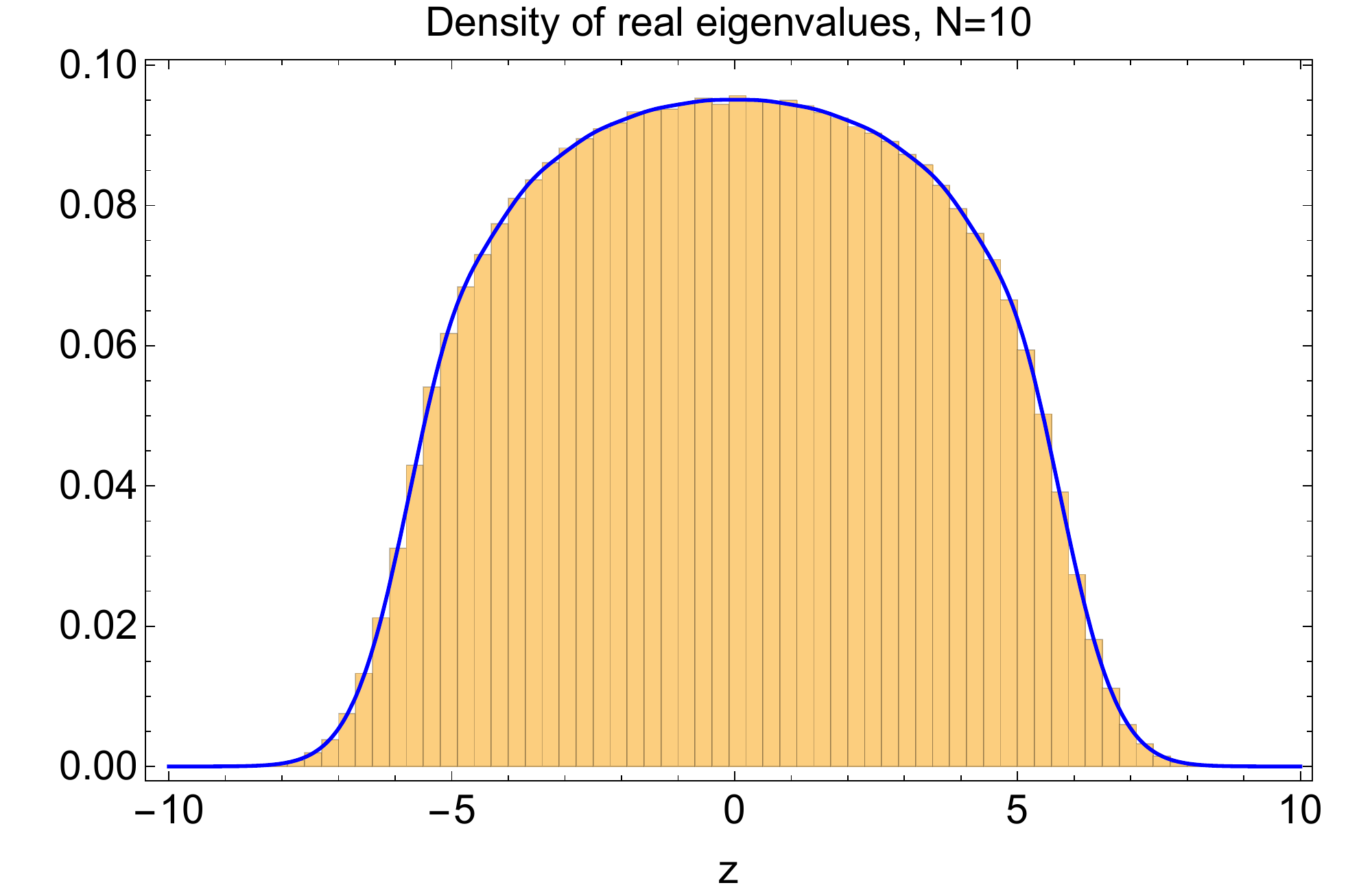}
\caption{Histograms of the density of real eigenvalues for the Real Elliptic ensemble with $\tau=0.9$ obtained by the direct diagonalization of  $10^6$ matrices of size $N=3$ (left) and  $2\cdot 10^5$ matrices of size $N=10$ (right). Blue solid lines represent the formula obtained by analytical (left) and numerical (right) integration of ${\cal P}(z,q)$~\eqref{eq:MainResult}. Formulas are rescaled so that the density is normalized to 1.
\label{Fig:2}
}
\end{figure}

\section{Proof of identity \eqref{eq:Identity}} \label{sec:App1}

We shall prove \eqref{eq:Identity} by induction.
\begin{proof}
The first step is trivial as this identity can be verified by substituting Hermite polynomials for low $N$. Let us assume that \eqref{eq:Identity} holds for $N-1$. Using formulas \eqref{eq:PNdef}-\eqref{eq:TNdef} it is easy to find the recurrence relations
\begin{align}
P_N&=NP_{N-1}+A_N, \\
R_N&=NR_{N-1}+B_N, \\
T_{N}&=NT_{N-1}+NA_N,
\end{align}
with
\begin{align}
A_N&= \tau^N[ (N+1)\He_N^2-N\He_{N+1}\He_{N-1}],  \label{eq:An}   \\
2B_{N} & = \tau^{N-1/2}[(N+1)\He_N\He_{N-1}-(N-1)\He_{N+1}\He_{N-2}, \label{eq:Bn}
\end{align}
where, for simplicity, we omitted the argument $\frac{z}{\sqrt{\tau}}$ of Hermite polynomials. These recursions allow us to rewrite the left-hand side of \eqref{eq:Identity} as
\begin{multline}
(N-1)[P_{N-1}-P_{N-2}(1+\tau-z^2)-(N-2)(2\tau^2+N-2)P_{N-3}-2z R_{N-2}+(1-\tau^2)(N-2)T_{N-3}]+ \\
A_{N}+(z^2-\tau)A_{N-1}-N(N-1)\tau^2A_{N-2}-2z B_{N-1}.
\end{multline}
The the expression in square brackets is zero by the induction assumption. Verification that the second line equals 0 relies on the substitution of \eqref{eq:An} and \eqref{eq:Bn} and the consecutive use of the three term recursion $\He_{N+1}(x)=x\He_{N}(x)-N\He_{N-1}(x)$.
\end{proof}

\end{appendices}

%


\begin{thebibliography}{99}
\bibitem{APS2008} G. Akemann, M. J. Phillips and H.-J. Sommers, \textit{Characteristic polynomials in real Ginibre ensembles}. J. Phys. A: Math. Theor. \textbf{42}, 012001 (2008).

\bibitem{ATTZ2019}
G. Akemann, R. Tribe, A. Tsareas, O. Zaboronski, \textit{On the determinantal structure of conditional overlaps for the complex Ginibre ensemble}. Random Matrices: Theory and Applications \textbf{9}, 2050015 (2020).

\bibitem{BNST2017} S. Belinschi, M. A. Nowak, R. Speicher and W. Tarnowski, \textit{Squared eigenvalue condition numbers and eigenvector correlations from the single ring theorem}.  J. Phys. A: Math. Theor. \textbf{50}, 105204 (2017).


\bibitem{BZ2018} F. Benaych-Georges and O. Zeitouni, \textit{Eigenvectors of non normal random matrices.}
Electron. Commun. Probab. {\bf 23}(70), 1-12 (2018).

\bibitem{BD2018}
P. Bourgade, G. Dubach, \textit{The distribution of overlaps between eigenvectors of Ginibre matrices}. Probability Theory and Related Fields \textbf{177} 397-464 (2020).

\bibitem{BGNTW2014}
Z. Burda, J. Grela, M. A. Nowak, W. Tarnowski and P. Warchoł, \textit{Dysonian dynamics of the Ginibre ensemble}.
Phys. Rev. Lett. \textbf{113}, 104102 (2014).

\bibitem{BSV2017}
 Z. Burda, B. J. Spisak and P. Vivo, \textit{Eigenvector statistics of the product of Ginibre matrices}.
Phys. Rev. E \textbf{95},  022134 (2017).

\bibitem{CM1998} J. T. Chalker and B. Mehlig, \textit{Eigenvector Statistics in Non-Hermitian Random Matrix Ensembles}.
Phys. Rev. Lett. \textbf{81}, 3367 (1998).

\bibitem{CR2018} N. Crawford and R. Rosenthal, \textit{Eigenvector correlations in the complex Ginibre ensemble}.
  arXiv:1805.08993 (2018).

\bibitem{DavyGenack18} M. Davy and A. Z. Genack, \textit{Selectively exciting quasi-normal modes in open disordered systems.}
 Nature Comm. {\bf 9},  4714 (2018).

\bibitem{DavyGenack19} M. Davy and A. Z. Genack, \textit{Probing nonorthogonality of eigenfunctions and its impact on transport through open systems}.   Phys. Rev. Research \textbf{1}, 033026 (2019).

\bibitem{D2019} G. Dubach, \textit{On eigenvector statistics in the spherical and truncated unitary ensembles}.
arXiv:1908.06713 [math.PR] (2019).

\bibitem{D2019a} G. Dubach, \textit{Symmetries of the quaternionic Ginibre ensemble}. Random Matrices: Theory and Applications (online ready) \url{https://doi.org/10.1142/S2010326321500131} (2020).

\bibitem{EKS1994}  A. Edelman, E. Kostlan and M. Shub, \textit{How many eigenvalues of a random matrix are real?}, J. Amer. Math. Soc. \textbf{7}, 247-267 (1994).

\bibitem{Efetov1997} K. B. Efetov, \textit{Directed quantum chaos}.
Phys. Rev. Lett. {\bf 79}, 491 (1997).

\bibitem{EKR2019} L. Erd\"{o}s, T. Kr\"{u}ger and D. Renfrew,
\textit{Randomly coupled differential equations with correlations.} arXiv:1908.05178 (2019).

\bibitem{FM2009}
P. J. Forrester and A. Mays, \textit{A Method to Calculate Correlation Functions for $\beta=1$ Random Matrices of Odd Size}. J. Stat. Phys. \textbf{134}, 443–462 (2009).

\bibitem{ForNag2008} P. J. Forrester and T. Nagao, \textit{Skew orthogonal polynomials and the partly symmetric real Ginibre ensemble.}
 J. Phys. A: Math. Theor. {\bf 41}, 375003 (2008).

  \bibitem{FS2000}
 K. M. Frahm, H. Schomerus, M. Patra and C. W. J. Beenakker, \textit{Large petermann factor in
chaotic cavities with many scattering channels}. Europhys. Lett. \textbf{49},  48 (2000).

\bibitem{F2002} Y. V. Fyodorov, \textit{Negative moments of characteristic polynomials of random matrices: Ingham–Siegel
integral as an alternative to Hubbard–Stratonovich transformation.} Nucl. Phys. B \textbf{621}(3), 643–674 (2002).

\bibitem{F2018} Y. V. Fyodorov, \textit{On statistics of bi-orthogonal eigenvectors in real and complex Ginibre ensembles: combining partial Schur decomposition with supersymmetry.}  Commun. Math. Phys. \textbf{363}, 579–603 (2018).

\bibitem{FK2016} Y. V. Fyodorov, B. Khoruzhenko, \textit{Nonlinear analogue of the May-Wigner instability transition}. PNAS 113 (25) 6827-6832 (2016).

\bibitem{FKS1997a} Y. V. Fyodorov, B. Khoruzhenko and H. -- J. Sommers, \textit{Almost-Hermitian Random Matrices: Eigenvalue Density in the
Complex  Plane}. Phys. Lett. A {\bf 226}, 46--52 (1997).

\bibitem{FKS1997b} Y. V. Fyodorov, B. Khoruzhenko and H. -- J. Sommers,
\textit{Almost-Hermitian Random Matrices: Crossover from Wigner-Dyson to
Ginibre Eigenvalue Statistics}. Phys. Rev. Lett. {\bf 79}, 557--560 (1997).

\bibitem{FKS1998} Y. V. Fyodorov, B. Khoruzhenko and H. -- J. Sommers, \textit{Universality in the random matrix spectra in the regime of weak non-Hermiticity}.
Ann. Inst. Henri Poincar\'e [Physique Theorique] \textbf{68}, 449-489 (1998).

\bibitem{FM2002}
Y. V. Fyodorov and B. Mehlig, \textit{Statistics of resonances and nonorthogonal eigenfunctions in
a model for single-channel chaotic scattering}. Phys. Rev. E \textbf{66}, 045202 (2002).

\bibitem{FS2012}
Y. V. Fyodorov and D. V. Savin, \textit{ Statistics of resonance width shifts as a signature of
eigenfunction non-orthogonality}. Phys. Rev. Lett. {\bf 108},  184101 (2012).

\bibitem{FS2003} Y. V. Fyodorov and  H. -- J. Sommers, \textit{Random matrices close to Hermitian or unitary: overview of methods and results.} J. Phys. A: Math. Theor.  {\bf 36}, 3303--3347 (2003).

\bibitem{FS2002} Y. V. Fyodorov and E. Strahov,  \textit{Characteristic polynomials of random Hermitian matrices and Duistermaat --
Heckman localization on non-compact K\"ahler manifolds}. Nucl. Phys. B \textbf{630}, 453--491 (2002).

\bibitem{G2017}
J. Grela, \textit{What drives transient behavior in complex systems?} Phys. Rev. E \textbf{96}, 022316 (2017).

\bibitem{GW2018}
J. Grela, P. Warchoł, \textit{Full Dysonian dynamics of the complex Ginibre ensemble}.
J. Phys. A: Math. Theor. \textbf{51}, 42 (2018).

\bibitem{GNetal2018}
E. Gudowska-Nowak, J. Ochab, D. Chialvo, M. A. Nowak and W. Tarnowski, \textit{From synaptic interactions to collective dynamics in random neuronal networks models: critical role of eigenvectors and transient behavior}. Neural Computation \textbf{32}, 395-423 (2020).

\bibitem{Jetal1999}
R. A. Janik, W. Noerenberg, M. A. Nowak, G. Papp and I. Zahed, \textit{Correlations of
eigenvectors for nonHermitian random matrix models}. Phys. Rev. E \textbf{60},  2699 (1999).

\bibitem{MBO2018} 
D. Mart\'i, N. Brunel and S. Ostojic, \textit{Correlations between synapses in pairs of neurons slow down dynamics in randomly connected neural networks}. Phys. Rev. E \textbf{97}, 062314 (2018). 

\bibitem{MC2000} B. Mehlig and J. T. Chalker,  \textit{Statistical properties of eigenvectors in non-Hermitian Gaussian random matrix ensembles}.
J. Math. Phys. \textbf{41}, 3233 (2000).

\bibitem{MetzNeri} I. Neri and F. L. Metz, \textit{Eigenvalue Outliers of Non-Hermitian Random Matrices with a Local Tree Structure.}
Phys. Rev. Lett. {\bf 117}, 224101 (2016); Erratum in: Phys. Rev. Lett. {\bf  118}, 019901 (2017).

\bibitem{NT2018}
M. A. Nowak and W. Tarnowski, \textit{Probing non-orthogonality of eigenvectors in non-Hermitian matrix models: diagrammatic approach}. JHEP  \textbf{2018} 152 (2018).


\bibitem{PSB2000}
M. Patra, H. Schomerus and C. W. J. Beenakker, \textit{Quantum-limited linewidth of a chaotic laser
cavity}. Phys. Rev. A \textbf{61}, 023810 (2000).

\bibitem{SFPB2000}
H. Schomerus, K. Frahm, M. Patra and C. W. J. Beenakker, \textit{Quantum limit of the laser line width
in chaotic cavities and statistics of residues of scattering matrix poles}. Physica A \textbf{278},
469 (2000).

\bibitem{TNV2019}
W. Tarnowski, I. Neri and P. Vivo, \textit{Universal transient behavior in large dynamical systems on networks}. Phys. Rev. Research \textbf{2}, 023333 (2020).

\bibitem{TE2005} L. N. Trefethen and M. Embree, \textit{Spectra and Pseudospectra. The Behavior of Nonnormal Matrices and Operators}. Princeton University Press, Princeton (2005).

\bibitem{W1965} J. H. Wilkinson, \textit{The Algebraic Eigenvalue Problem}. Oxford University Press, Oxford (1965).

 \bibitem{WaltersStarr2015} M. Walters and S. Starr, \textit{A note on mixed matrix moments for the complex Ginibre ensemble.}
 J. Math. Phys. {\bf 56}, 013301 (2015).



\end{thebibliography}
\end{document}